\documentclass{article}
\usepackage{fullpage,latexsym,amsthm,amsfonts,amssymb,amsmath,amsthm}
\usepackage{tikz}
\usepackage{appendix}
\pdfoutput=1 

\def\01{\{0,1\}}

\newcommand{\sgn}{\mbox{\rm sgn}}

\newcommand{\connect}{\mbox{\sc Connected}}
\newcommand{\matching}{\mbox{\sc Matching}}
\newcommand{\trianglefind}{\mbox{\sc Triangle}}
\newcommand{\euler}{\mbox{\sc Euler}}
\newcommand{\planar}{\mbox{\sc Planar}}
\newcommand{\bipart}{\mbox{\sc Bipartite}}
\newcommand{\DISJ}{\mbox{\sc Disj}}
\newcommand{\IP}{\mbox{\sc IP}}
\newcommand{\OR}{\mbox{\sc OR}}
\newcommand{\DET}{\mathrm{DET}}

\newtheorem{theorem}{Theorem}
\newtheorem{myclaim}{Claim}

\begin{document}

\title{New bounds on the classical and quantum communication complexity of some graph properties\thanks{Most of this work was conducted at the Centre for Quantum Technologies (CQT) in Singapore, and partially funded by the Singapore Ministry of Education and the National Research Foundation. Research partially supported by the European Commission IST project Quantum Computer Science (QCS) 255961, by the CHIST-ERA project DIQIP,
by Vidi grant 639.072.803 from the Netherlands Organization for Scientific Research (NWO),
by the French ANR programs under contract ANR-08-EMER-012 (QRAC project) and ANR-09-JCJC-0067-01 (CRYQ project), by the French MAEE STIC-Asie program FQIC, and by the Hungarian Research Fund (OTKA).}}
\author{G\'abor Ivanyos\thanks{Computer and Automation Research Institute
of the Hungarian Academy of Sciences, Budapest, Hungary.
Email: {\tt Gabor.Ivanyos@sztaki.hu}}
\and
Hartmut Klauck\thanks{CQT and NTU Singapore. Email: {\tt hklauck@gmail.com}}
\and
Troy Lee\thanks{CQT Singapore. Email: {\tt troyjlee@gmail.com}}
\and
Miklos Santha\thanks{CNRS - LIAFA, Universit\'e Paris Diderot, Paris, France and Centre for Quantum Technologies, National University of Singapore. Email: {\tt santha@liafa.jussieu.fr}}
\and
Ronald de~Wolf\thanks{CWI and University of Amsterdam. Email: {\tt rdewolf@cwi.nl}}}
\date{}
\maketitle

\begin{abstract}
We study the communication complexity of a number of graph properties
where the edges of the graph $G$ are distributed between Alice and Bob (i.e.,
each receives some of the edges as input).
Our main results are:
\begin{itemize}
\item An $\Omega(n)$ lower bound on the quantum communication complexity of deciding whether 
an $n$-vertex graph $G$ is connected,
nearly matching the trivial classical upper bound of $O(n\log n)$ bits of communication.
\item A deterministic upper bound of $O(n^{3/2}\log n)$ bits for deciding if a bipartite
graph contains a perfect matching, and a quantum lower bound of $\Omega(n)$ for this problem.
\item A $\Theta(n^2)$ bound for the randomized communication complexity of deciding if a graph has an Eulerian tour, and a $\Theta(n^{3/2})$ bound for
the quantum communication complexity of this problem.
\end{itemize}
The first two quantum lower bounds are obtained by exhibiting a reduction from the $n$-bit Inner Product problem
to these graph problems, which solves an open question of Babai, Frankl and Simon~\cite{bfs:classes}. The third quantum lower bound comes
from recent results about the quantum communication complexity of composed functions.
We also obtain essentially tight bounds for the quantum communication
complexity of a few other problems, such as deciding if $G$ is triangle-free, or if $G$ is bipartite, as well as computing the determinant of a distributed matrix.
\end{abstract}

\section{Introduction}

Graphs are among the most basic discrete structures, and deciding whether graphs
have certain properties (being connected, containing a perfect matching,
being 3-colorable, \ldots) is among the most basic computational tasks.
The complexity of such tasks has been studied in a number of different settings.

Much research has gone into the \emph{query complexity} of graph properties, most of it
focusing on the so-called Aandera-Karp-Rosenberg conjecture.
Roughly speaking, this says that all monotone graph properties have
query complexity $\Omega(n^2)$. Here the vertex set is $[n]=\{1,\ldots,n\}$ and input graph $G=([n],E)$
is given as an adjacency matrix whose entries can be queried. 
This conjecture is proved for deterministic algorithms~\cite{rivest&vuillemin},
but open for randomized ones~\cite{hajnal:graphprop,chakrabarti&khot:graphprop}.

Less---but still substantial---effort has gone into the study of the \emph{communication complexity}
of graph properties~\cite{papadimitriou&sipser:cc,bfs:classes,hmt:ccgraphprop,pudlak&duris:planarity}.
Here the edges of $G$ are distributed over two parties, Alice and Bob.
Alice receives set of edges $E_A$, Bob receives set $E_B$ (these sets may overlap), and the goal
is to decide with minimal communication whether the graph $G=([n],E_A\cup E_B)$ has a certain property.

In this paper we obtain new bounds for the communication complexity of a number of graph properties,
both in the classical and the quantum world.
Our main results are:
\begin{itemize}
\item An $\Omega(n)$ lower bound on the quantum communication complexity of deciding whether $G$ is connected,
nearly matching the trivial classical upper bound of $O(n\log n)$ bits.
\item Hajnal et al.~\cite{hmt:ccgraphprop} state as an open problem to determine the communication complexity of deciding if a bipartite
graph contains a perfect matching (i.e., a set of $n/2$ vertex-disjoint edges).
We prove a deterministic upper bound of $O(n^{3/2}\log n)$ bits for this, and a quantum lower bound of $\Omega(n)$.
\item For the problem of deciding if a graph contains an Eulerian tour we show that the quantum communication complexity is $\Theta(n^{3/2})$, whereas the randomized communication complexity is $\Theta(n^2)$.
\end{itemize}
Our quantum lower bounds for the first two problems are proved by reductions from the hard inner product problem,
which is $\IP_n(x,y)=\sum_{i=1}^n x_i y_i$ mod 2.
Babai et al.~\cite[Section~7]{bfs:classes} showed how to reduce the disjointness problem 
($\DISJ_n(x,y)=1$ iff $\sum_{i=1}^n x_iy_i=0$) to these graph problems,
but left reductions from inner product as an open problem (they did reduce inner product to
a number of other problems~\cite[Section~9]{bfs:classes}).
In the classical world this doesn't make much difference since both \DISJ\ and \IP\ require $\Omega(n)$ communication (the tight lower bound for \DISJ\ was proved only after~\cite{bfs:classes} in~\cite{ks:disj}).  However, in the quantum world \DISJ\ is quadratically easier than \IP, so reductions from \IP\ give much stronger lower bounds in this case.

While investigating the communication complexity of graph properties is interesting in its own right, there have also been applications of lower bounds for such problems.
For instance, communication complexity arguments have recently been used to show new and tight lower bounds for several graph problems in distributed computing in~\cite{dassarma:distributed}. These problems include approximation and verification versions of classical graph problems like connectivity, $s$-$t$ connectivity, and bipartiteness. In their setting processors see only their local neighborhood in a network. \cite{dassarma:distributed} use reductions from \DISJ\ to establish their lower bounds. Subsequently some of these results have been generalized to the case of quantum distributed computing \cite{eknp:distributed}, employing for instance the new reductions from \IP\ given in this paper, which in the quantum case establish larger lower bounds than the previous reductions from \DISJ.

\section{Preliminaries}

We assume familiarity with communication complexity,
referring to~\cite{kushilevitz&nisan:cc} for more details about classical communication complexity
and~\cite{wolf:qccsurvey} for quantum communication complexity (for information about
the quantum model beyond what's provided in~\cite{wolf:qccsurvey}, see~\cite{nielsen&chuang:qc}).

Given some communication complexity problem $f:X\times Y\rightarrow R$
we use $D(f)$ to denote its classical deterministic communication complexity,
$R_2(f)$ for its private-coin randomized communication complexity with error
probability $\leq 1/3$, and $Q_2(f)$ for its private-coin quantum
communication complexity with error $\leq 1/3$.
Our upper bounds for the quantum model do not require prior shared entanglement;
however, all lower bounds on $Q_2(f)$ in this paper also apply to the
case of unlimited prior entanglement.

Among others we consider two well-known communication complexity problems,
with $X=Y=\01^n$ and $R=\01$.  For $x,y \in \01^n$ we define $x\wedge y \in \01^n$ as the bitwise
AND of $x$ and $y$, and $|x|=|\{i \in [n]: x_i = 1\}|$ as the Hamming weight of $x$.
\begin{itemize}
\item Inner product: $\IP_n(x,y)= |x\wedge y|$ mod 2.  The quantum communication
complexity of this problem is $Q_2(\IP_n)=\Theta(n)$~\cite{kremer:thesis,cdnt:ip}
(in fact even its \emph{unbounded-error} quantum communication complexity is linear~\cite{forster:probcc}).
\item Disjointness: $\DISJ_n(x,y)=1$ if $|x\wedge y|=0$, and $\DISJ_n(x,y)=0$ otherwise.
Viewing $x$ and $y$ as the characteristic vectors of subsets of $[n]$, the task is to decide whether these sets are disjoint.
It is known that $R_2(\DISJ_n)=\Theta(n)$~\cite{ks:disj,razborov:disj}
and $Q_2(\DISJ_n)=\Theta(\sqrt{n})$~\cite{BuhrmanCleveWigderson98,aaronson&ambainis:searchj,razborov:qdisj}.
In fact, the Aaronson-Ambainis protocol~\cite{aaronson&ambainis:searchj}
can find an $i$ such that $x_i=y_i=1$ (if such an $i$ exists),
using an expected number of $O(\sqrt{n})$ qubits of communication.
This saves a log-factor compared to the more straightforward distributed implementation
of Grover's algorithm in~\cite{BuhrmanCleveWigderson98}.
\end{itemize}

\section{Reduction from Parity}

We begin with a reduction from the $n$-bit Parity problem to
the connectedness of a $2n$-vertex graph in the model of \emph{query} complexity.  This reduction was used
by D\"urr et al.~\cite[Section~8]{dhhm:graphproblemsj}, who attribute it to Henzinger and
Fredman~\cite{henzinger&fredman:connectivity}.  The same reduction can be used to reduce Parity to
determining if an $n$-by-$n$ bipartite graph contains a perfect matching.
Our hardness results for communication complexity in later sections follow by
means of simple gadgets to transfer this reduction from the query world to the communication world.

\begin{myclaim}
For every $z \in \{0,1\}^n$ there is a graph $G_z$ with $2n$ vertices
(where for each possible edge, its presence or absence just depends on one of the bits of $z$), such that if the parity of $z$
is odd then $G_z$ is a cycle of length $2n$, and if the parity of $z$ is even then $G_z$ is the disjoint
union of two $n$-cycles.
\label{claim:queryred}
\end{myclaim}

\begin{proof}
We construct a graph $G$ with $2n$ vertices, arranged in two rows of $n$ vertices each.  We will
label the vertices as $t_i$ and $b_i$ for $i \in [n]$ indicating if it is in the top row or the bottom row.
For $i\in[n-1]$, if $z_i=0$ then add edges $\{t_i,t_{i+1}\}$ and $\{b_i,b_{i+1}\}$;
if $z_i=1$ then add $\{t_i,b_{i+1}\}$ and $\{b_i,t_{i+1}\}$.  For $i=n$ make the same connections
with vertex $1$, wrapping around.  See Figure~\ref{figparitytoconnect} for illustration.
If the parity of $z$ is odd then the resulting graph $G$ will be one $2n$-cycle, and if the parity is even then it will be two $n$-cycles.
\end{proof}

\begin{figure}[hbt]
\centering
\includegraphics[scale=.6]{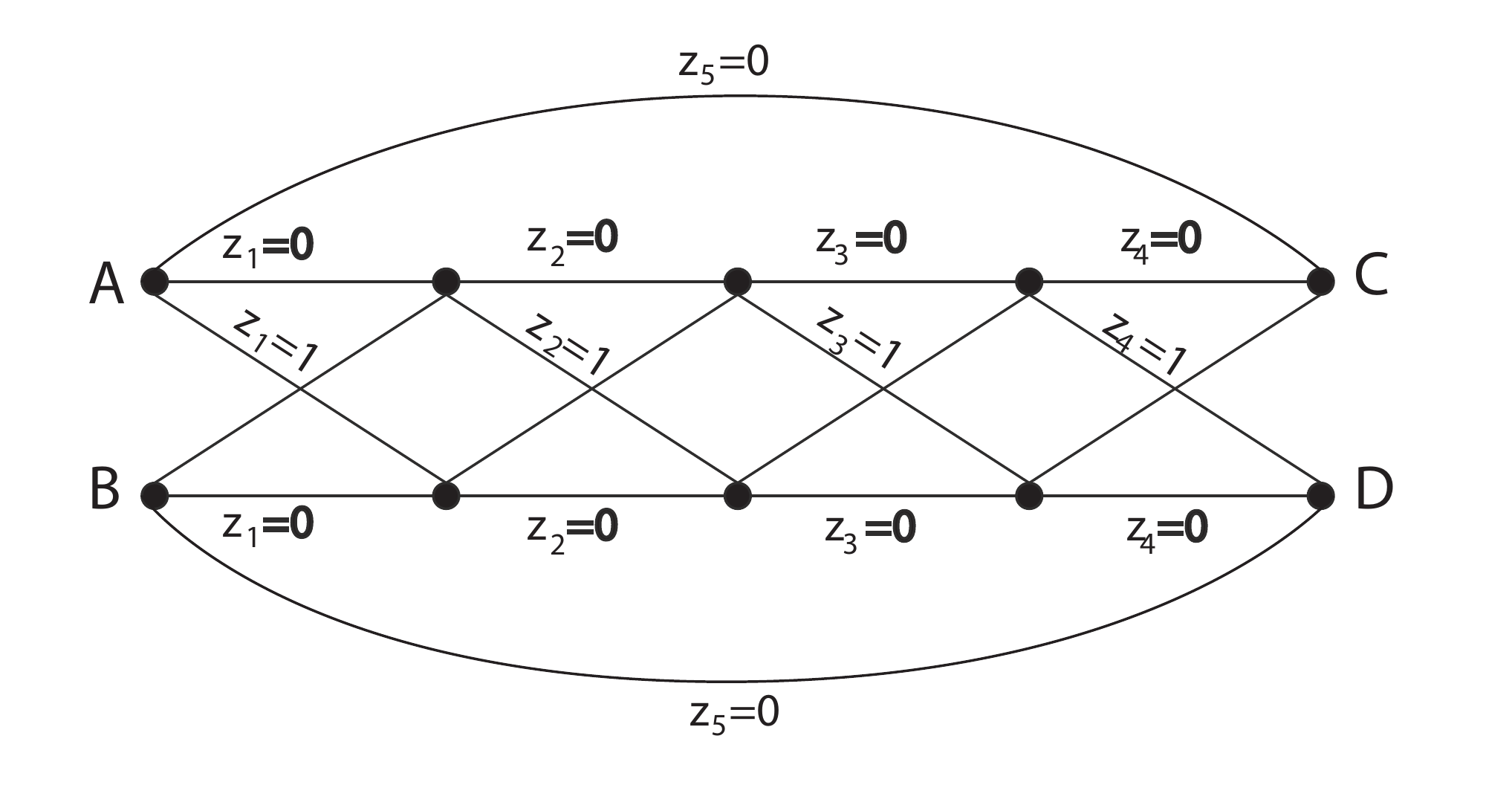}
\caption{The string $z$ determines the edges present in $G_z$.
If $z_5=1$ there are edges connecting A with D, and B with C (omitted for clarity).
When the parity of $z$ is odd, the graph is a $2n$-cycle, and when it is even the graph is the
disjoint union of two $n$-cycles.}\label{figparitytoconnect}
\end{figure}

\section{Connectedness}\label{secconnect}

We first focus on the communication complexity of deciding whether a graph $G$ is connected or not.
Denote the corresponding Boolean function for $n$-vertex graphs by $\connect_n$
(we sometimes omit the subscript when it's clear from context).
Note that it suffices for Alice and Bob to know the connected components of their graphs;
additional information about edges within their connected components is redundant for deciding connectedness.
Hence the ``real'' input length is $O(n\log n)$ bits, which of course implies the upper bound $D(f)=O(n\log n)$.
Hajnal et al.~\cite{hmt:ccgraphprop} showed a matching lower bound for $D(f)$.
As far as we know, extending this lower bound to $R_2(\connect)$ is open.
The best lower bound known is $R_2(\connect)=\Omega(n)$ via a reduction from $\DISJ_n$~\cite{bfs:classes}.
Since $\DISJ$ is quadratically easier for quantum communication than for classical communication,
the reduction from $\DISJ_n$ only implies a quantum lower bound $Q_2(\connect)=\Omega(\sqrt{n})$.
We now improve this by giving a reduction from $\IP_n$, answering an open question from~\cite{bfs:classes}.
Since we know $Q_2(\IP_n)=\Omega(n)$, this will imply $Q_2(\connect)=\Omega(n)$, which is tight up to the log-factor.

We modify the graph from Claim~\ref{claim:queryred}
originally used in the context of query complexity to give a reduction from inner product to
connectedness in the communication world.
\begin{theorem}
$\Omega(n)\leq  Q_2(\connect_n)\leq  D(\connect_n)\leq O(n\log n)$.
\end{theorem}
\begin{proof}
Let $x\in\01^n$ and $y\in\01^n$ be Alice and Bob's inputs, respectively.
Set $z=x\wedge y$, 
then the parity of $z$ is $
\IP_n(x,y)$.
We define a graph $G$ which is a modification of the graph $G_z$ from Claim~\ref{claim:queryred} by distributing
its edges over
Alice and Bob,
in such a way that if $\IP_n(x,y)=1$ (i.e., $|z|$ is odd) then the resulting
graph is a $2n$-cycle, and if $\IP_n(x,y)=0$ (i.e., $|z|$ is even)
then $G$ consists of two disjoint $n$-cycles, and therefore is not connected.
To do that we replace every edge with a ``gadget'' that adds two extra vertices. Formally,
we will have the $2n$ vertices $t_i, b_i$, and $8n$ new vertices
$k_i^{tt}, k_i^{bb}, k_i^{tb}, k_i^{bt}, \ell_i^{tt}, \ell_i^{bb}, \ell_i^{tb}, \ell_i^{bt} ,$
for $i\in [n]$.
See Figure~\ref{fig:conGadget} for a picture of the gadgets.

We describe the gadget corresponding to the $i$th horizontal edge on the top. It involves the vertices
$t_i,  k_i^{tt}, \ell_i^{tt}, t_{i+1}$ and depends only on $x_i$ and $y_i$. The gadget corresponding to the $i$th
horizontal bottom edge is isomorphic but defined on vertices $b_i,  k_i^{bb}, \ell_i^{bb}, b_{i+1}$.
If $x_i = 0$ then $\{t_i,  k_i^{tt}\} \in E_A$, and  if $y_i=0$ then $\{t_i,  \ell_i^{tt}\} \in E_B$.
Independently of the value of $x_i$, the edges
$\{k_i^{tt}, t_{i+1}\}$ and $\{\ell_i^{tt}, t_{i+1}\}$ are in $E_A$.
Note that this gadget is connected iff $x_i y_i=0$.

Now we describe the gadget corresponding to the $i$th diagonal edge $\{t_i, b_{i+1}\}$, the gadget corresponding to
$\{b_i, t_{i+1}\}$ is isomorphic to this one on the appropriate vertex set. If $x_1 = 1$ then $\{t_i, \ell_i^{tb}\} \in E_A$, if
$y_i = 0$ then $\{k_i^{tb}, \ell_i^{tb}\} \in E_B$, and if $y_i =1$ then $\{\ell_i^{tb}, t_{i+1}\} \in E_B$. Finally
$\{t_i, k_i^{tb}\} \in E_A$ no matter what $x_i$ is.
Note that this gadget is connected iff $x_i y_i=1$.

In total the resulting graph $G$ will have $10n$ vertices, and disjoint sets $E_A$ and $E_B$ of $O(n)$ edges.
If $\IP_n(x,y)=1$ then the graph consists of one cycle of length $4n$, with a few extra vertices attached to it.
If $\IP_n(x,y)=0$ then the graph consists of two disjoint cycles of length $2n$ each, again
with a few extra vertices attached to them.
(Observe that $\ell_i^{tb}$ is always connected to $t_i$ or to $t_{i+1}$ even when $x_i=y_i=0$).
Accordingly, a protocol that can compute $\connect$ on this graph computes $\IP_n(x,y)$,
which shows $Q_2(\IP_n)\leq Q_2(\connect_{10n})$.

Our gadgets are slightly more complicated than strictly necessary, to ensure the sets of edges $E_A$ and $E_B$ are disjoint.
This implies that the lower bound holds even for that special case.
Note that the lower bound even holds for \emph{sparse} graphs, as $G$ has $O(n)$ edges.
\end{proof}

\begin{figure}[hbt]
\centering
\includegraphics[scale=.6]{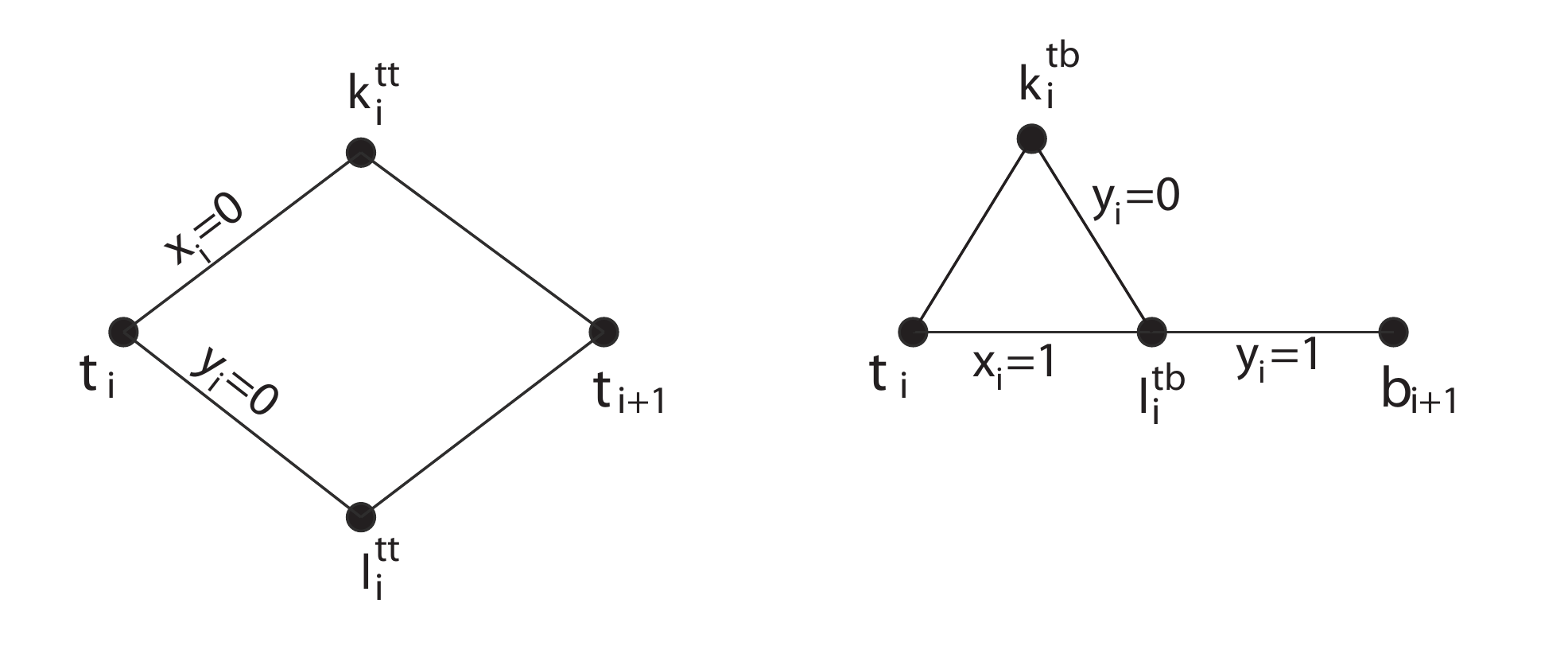}
\caption{Two gadgets used to modify the reduction from Parity in the query complexity model to one for
Inner Product in the communication complexity model.  On the left the gadget used to replace the
top $z_i=0$ edge in the graph from Figure~\ref{figparitytoconnect}, and on the right the gadget for the
diagonal top to bottom $z_i=1$ edge.}\label{fig:conGadget}
\end{figure}

\section{Matching}

The second graph problem we consider is deciding whether an $n\times n$ \emph{bipartite} graph
$G$ contains a perfect matching.
We denote this problem by $\matching_n$.
First, we show that the above reduction from \IP\ can be modified to also work for \matching.

\begin{theorem}\label{theorem:matchlower}
$Q_2(\matching_n)\geq\Omega(n)$.
\end{theorem}

\begin{proof}
Let $x\in\01^n$ and $y\in\01^n$ be respectively the inputs of Alice and Bob.
As previously, we set $z=x\wedge y$, 
and observe again that the parity of $z$ is $
\IP_n(x,y)$.
We go back to the query world and the $2n$-vertex graph $G_z$ of Claim~\ref{claim:queryred}.
Assume $n$ is odd.
Then in case the parity of $z$ is odd, $G_z$ is a cycle of even length
$2n$ and so has a perfect matching.  On the other hand, in case the parity of $z$ is even, $G_z$
consists of two odd cycles and so has no perfect matching.

Now we again use gadgets to transfer this idea to a reduction from inner product to matching in
the communication complexity setting.  For simplicity we first describe the reduction where the
edge sets of Alice and Bob can overlap.  We then explain a modification to make them
disjoint.

The vertices of the graph $G$ will consist of the $2n$ vertices $t_i, b_i$ as in Figure~
\ref{figparitytoconnect} with the addition of $4n$ new vertices $k_i^{t}, k_i^{b}, \ell_i^{t}, \ell_i^{b}$
for $i\in [n]$.  For every $i$ there is a unique gadget on vertex set $\{t_i, b_i, k_i^{t}, k_i^{b}, \ell_i^{t}, \ell_i^{b},
t_{i+1}, b_{i+1}\}$.
The edges  $\{k_i^{t}, \ell_i^{b}\}$ and $\{k_i^{b}, \ell_i^{t}\}$ are always present in the graph,
and will be given to Alice.   If $x_i=0$ then we give Alice the edges $\{t_i, t_{i+1}\}$ and
$\{b_i, b_{i+1}\}$.
If $y_i=0$ we
do the same thing for Bob (this is where edges may overlap).
If $x_i=1$ we give Alice the edges $\{t_i, k_i^{t}\}$ and $\{b_i, k_i^{b}\}$.
If $y_i=1$ we give Bob the edges $\{t_{i+1}, \ell_i^{t}\}$ and $\{b_{i+1}, \ell_i^{b}\}$.
This is illustrated in Figure~\ref{figmatchGadget}.

Now in case the parity of $z$ is odd, we will have a cycle of even length, with possibly some
additional disjoint edges and attached paths of length two.  Thus there will be a perfect matching.
In case the parity of $z$ is even, we will have two odd cycles, and again  some additional disjoint
edges or attached paths of length two.  Suppose, by way of contradiction, that there is a perfect 
matching in this case.  In case $x_i y_i=0$, this matching must include the edge $\{k_i^{t}, \ell_i^{b}\}$, 
since at least one of these vertices has degree one, and similarly for $\{k_i^{b}, \ell_i^{t}\}$.  
Thus a perfect matching in this case gives a perfect matching of two odd cycles, a contradiction.

To make the edge sets disjoint, we replace horizontal edges between vertex $i$ and $i+1$ by
the gadget in the left of Figure~\ref{figmatchGadget}.
It can be seen that this does not change the properties used in the reduction.
\begin{figure}[hbt]
\centering
\includegraphics[scale=.6]{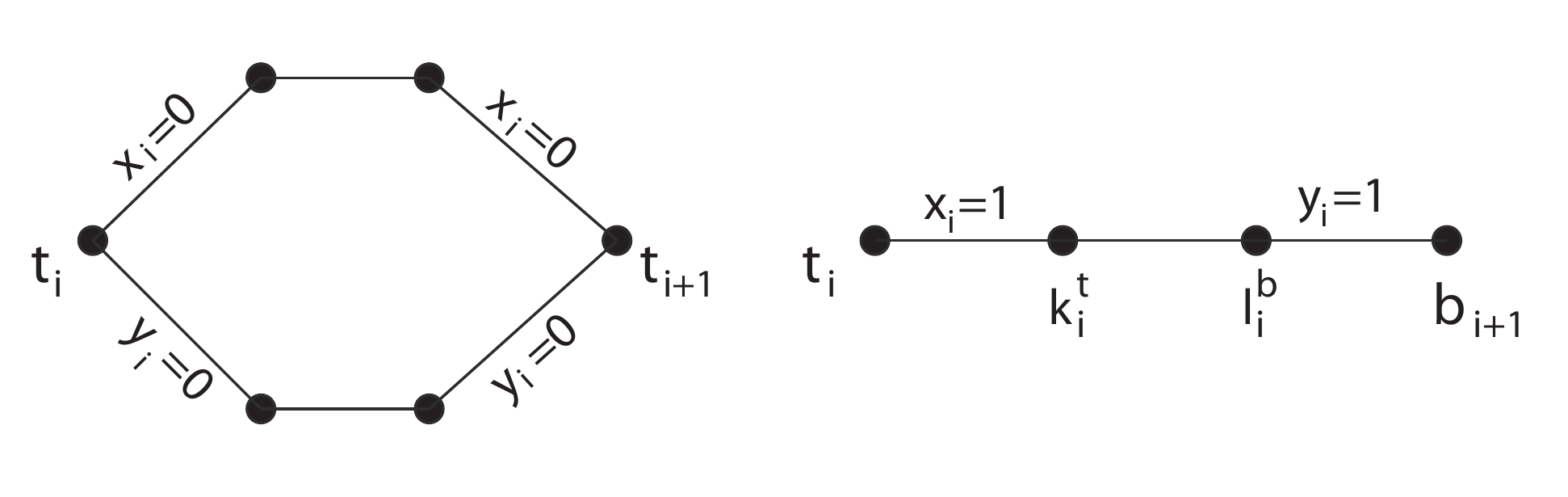}
\caption{Two gadgets used to modify the reduction from Parity in the query complexity model
to one for matching in the communication complexity model.  On the right is the gadget
for the top to bottom diagonal $z_i=1$ edge.  On the left, the gadget used to replace the
top $z_i=0$ horizontal edge in the graph from Figure~\ref{figparitytoconnect} such that
Alice and Bob receive disjoint sets of edges.}\label{figmatchGadget}
\end{figure}
\end{proof}

Second, we show a non-trivial deterministic upper bound $D(\matching_n)=O(n^{3/2}\log n)$
by implementing a distributed version of the famous Hopcroft-Karp
algorithm for finding a maximum-cardinality matching~\cite{hopcroft&karp:matching}.
Let us first explain this algorithm in the standard non-distributed setting.
The algorithm starts with an empty matching $M$, and in each iteration grows the size of $M$ until it can no
longer be increased. It does this by finding, in each iteration, many \emph{augmenting paths}.
An augmenting path, relative to a matching $M$, is a path $P$ of odd length that starts and ends at ``free'' (= unmatched in $M$) vertices,
and alternates non-matching with matching edges.  Note that the symmetric difference of $M$ and $P$ is another matching,
of size one greater than $M$.
Each iteration of the Hopcroft-Karp algorithm does the following
(using the notation of~\cite{hopcroft&karp:matching},
 we call the vertex sets of the bipartition $X$ and $Y$, respectively).
\begin{enumerate}
\item A breadth-first search (BFS) partitions the vertices of the graph into layers. The free vertices in
$X$ are used as the starting vertices of this search, and form the initial layer of the partition. The traversed edges are required to alternate between unmatched and matched. That is, when searching for successors from a vertex in
$X$, only unmatched edges may be traversed, while from a vertex in $Y$ only matched edges may be traversed. The search terminates at the first layer $k$ where one or more free vertices in
$Y$ are reached.
\item All free vertices in $Y$ at layer $k$ are collected into a set $F$. That is, a vertex $v$ is put into $F$ iff it ends a shortest augmenting path
(i.e., one of length~$k$).
The algorithm finds a maximal set of vertex-disjoint augmenting paths of length $k$. This set may be computed by depth-first search (DFS) from $F$ to the free vertices in
$X$, using the BFS-layering to guide the search: the DFS is only allowed to follow edges that lead to an unused vertex in the previous layer, and paths in the DFS tree must alternate between unmatched and matched edges.
Once an augmenting path is found that involves one of the vertices in $F$, the DFS is continued from the next starting vertex.
After the search is finished, each of the augmenting paths found is used to enlarge~$M$.
\end{enumerate}
The algorithm stops when a new iteration fails to find another augmenting path, at which point the current $M$ is a maximal-cardinality matching.
Hopcroft and Karp showed that this algorithm finds a maximum-cardinality matching using $O(\sqrt{n})$ iterations.
Since each iteration takes time $O(n^2)$ to implement, the overall time complexity is $O(n^{5/2})$.

Now consider what happens in a distributed setting, where Alice and Bob each have some of the edges of $G$.
In this case, one iteration of the Hopcroft-Karp algorithm can be implemented
by having each party perform as much of the search as possible within their
graph, and then communicate the relevant vertices and edges to the other.
To be more specific, the BFS is implemented as follows. For each level,
first Alice scans the vertices on the given level and lists the set of
vertices which belong to the next level due to edges seen by Alice,
and then Bob lists the remaining vertices of the next level. When doing
a DFS, first Alice goes forward as much as possible, then Bob follows.
If Bob cannot continue going forward he gives the control back to Alice
who will step back. Otherwise Bob goes forward as much as he can and
then gives the control back to Alice who can either step back or continue
going forward. During both types of search, when a new
vertex is discovered Alice or Bob communicates the vertex as well as
the edge leading to the new vertex. (Note that both the BFS and the DFS
give algorithms of communication cost $\Theta(n\log n)$
for the constructive version of connectivity.)

Since each vertex needs to be communicated at most once per iteration,
implementing one iteration thus takes $O(n\log n)$ bits of communication.
Since there are $O(\sqrt{n})$ iterations, the whole procedure can be implemented using $O(n^{3/2}\log n)$ bits of communication.
Finding the maximum-cardinality matching of course suffices for deciding if $G$ contains a perfect matching,
so we get the same upper bound on $D(\matching_n)$ (strangely, we don't know anything better when we allow randomization and quantum communication).
We have proved:

\begin{theorem}
$D(\matching_n) \leq O(n^{3/2}\log n)$.
\end{theorem}

\subsection{Reducing communication by distributing Lov\'asz's~algorithm?}
In the usual setting of computation (not communication), Lov\'asz~\cite{lovasz:perfectmatch} gave a very elegant randomized method
to decide whether a bipartite graph contains a perfect matching in matrix-multiplication time.
Briefly, it works as follows.  The determinant of an $n\times n$ matrix $A$~is
$$
\det(A)=\sum_{\sigma\in S_n}\sgn(\sigma)\prod_{i=1}^n A_{i,\sigma(i)}.
$$
This $\det(A)$ is a degree-$n$ polynomial in the matrix entries.
Suppose we fix the $A_{ij}$ that equal~0, and replace the other $A_{ij}$ by variables $x_{ij}$.
This turns $\det(A)$ into a polynomial $p(x)$ of degree $n$ in (at most) $n^2$ variables $x_{ij}$.
Note that the monomial $\prod_{i=1}^n x_{i,\sigma(i)}$ vanishes iff at least one of the $A_{i,\sigma(i)}$ equals~0.
Hence a graph $G$ has no perfect matching iff the polynomial $p(x)$ derived from its bipartite adjacency matrix $A$ is identically equal to~0.
Testing whether a polynomial $p$ is identically equal to~0 is easy to do with a randomized algorithm:
randomly choose values for the variables $x_{ij}$ from a sufficiently large field, and compute the value of the polynomial $p(r)$.
If $p\equiv 0$ then $p(r)=0$, and if $p\not\equiv 0$ then $p(r)\neq 0$ with high probability by
the Schwartz-Zippel lemma~\cite{schwartz:probabilistic,zippel:probabilistic}.
Since $p(x)$ is the determinant of an $n\times n$ matrix, which can be computed in matrix-multiplication time $O(n^\omega)$,%
\footnote{The current best bound is $\omega\in[2,2.373)$~\cite{coppersmith&winograd:matrixmult,stothers:phd,vassilevska:matrixmult}.}
we obtain the same upper bound on the time needed to decide with high probability whether a graph contains a perfect matching.

One might hope that a distributed implementation of Lov\'asz's algorithm could improve the above 
communication protocol for matching, using randomization and possibly even quantum communication.
Unfortunately this does not work, because it turns out that computing the determinant of an $n\times n$
matrix whose $n^2$ entries are distributed over Alice and Bob, takes $\Omega(n^2)$ qubits of communication.
In fact, even deciding whether the determinant equals~0 modulo~2 takes $\Omega(n^2)$ qubits of communication.
We show this by reduction from $\IP_{n^2}$.  Let $\DET_n$ be the communication problem where
Alice is given an $n$-by-$n$ Boolean matrix $X$, Bob an $n$-by-$n$ Boolean matrix $Y$, and the desired output is
$\det(X \wedge Y)$, where $X \wedge Y$ is the bitwise AND of $X$ and $Y$.

\begin{theorem}
$\Omega(n^2) \le Q_2(\DET_n)$.
\end{theorem}

\begin{proof}
As before, we first explain a reduction in the query world from Parity of $n^2$ bits to computing the
determinant of a $(2n+2) \times (2n+2)$ matrix.  The basic idea of the proof goes back to Valiant~\cite{valiant:det}.
Say that we want to compute the parity of the bits of an $n^2$-bit string $Z$,
and arrange the bits of $Z$ into an $n$-by-$n$ matrix.  We construct a directed bipartite graph
$G_Z$ with $2n+2$ vertices, $n+1$ on each side (we will refer to these as left-hand side and
right-hand side).  Label the vertices on the left-hand side
as $t$ and $\ell_i$ for $i \in [n]$,
and those on the right-hand side as $s$ and $r_i$ for $i \in [n]$.  For every $i \in [n]$, we add the edges
$(s, \ell_i)$ and $(r_i, t)$.
For every $(i,j)$ with $Z(i,j)=1$ we
put an edge $(\ell_i, r_j)$.  Finally we put the edge $(t,s)$, and self-loops are added to all vertices but $s$ and $t$.
\begin{myclaim}
\label{detqueryclaim}
$\det(G_Z)=-|Z|$.
\end{myclaim}
\begin{proof}
Note that
\[
\det(G_Z)=\sum_\sigma (-1)^{\chi(\sigma)} \prod_i G_Z(i, \sigma(i)).
\]
Consider a permutation that contributes to this sum.  In this case, $\sigma(\ell_i)=r_j$ for some
$i,j$ for which $Z(i,j)=1$.  We then must have $\sigma(r_j)=t, \sigma(t)=s, \sigma(s)=\ell_i$ and
that $\sigma$ fixes all other vertices.  The sign of this permutation is negative, and we get such
a contribution for every $i,j$ such that $Z(i,j)=1$.
\end{proof}
Now again we transfer this reduction to the communication complexity setting by means of
a gadget.  Say that Alice has $X$, an $n$-by-$n$ matrix and similarly Bob has $Y$ and they
want to compute $|X \wedge Y| \bmod 2$.  We will actually count the number of zeros in
$X \wedge Y$, which clearly then allows us to know the number of ones and so the parity.

We give Alice the set of edges $E_A$ and Bob the set of edges $E_B$.  Unlike in the previous reductions, in this case $E_A$ and $E_B$ will not be disjoint (we do not know how to do the reduction with disjoint $E_A,E_B$).
Put $(s,\ell_i), (\ell_i, \ell_i) \in E_A$ for all $i \in [n]$ and
similarly $(t,r_i), (r_i,r_i) \in E_B$ for
all $i\in [n]$.  For all $(i,j)$ where $X(i,j)=0$ put $(\ell_i,r_j) \in E_A$, and similarly for all
$(i,j)$ where $Y(i,j)=0$ put $(\ell_i,r_j) \in E_B$.  Thus in $E_A \cup E_B$ there is an edge
$(\ell_i,r_j)$ if and only if $X(i,j) Y(i,j)=0$.  Thus by Claim~\ref{detqueryclaim} from the determinant
of the graph with edges $E_A \cup E_B$ we can determine the number of zeros in $X \wedge Y$.%
\end{proof}

\begin{figure}[hbt]
\centering
\includegraphics[scale=.6]{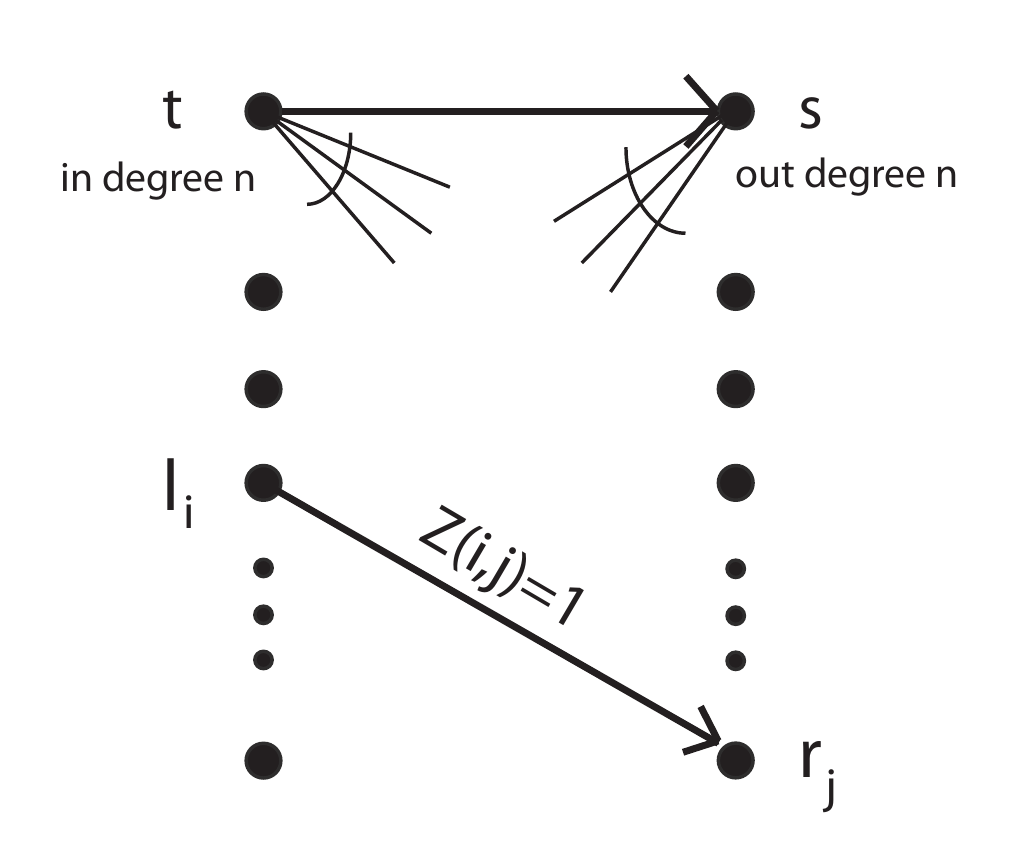}
\caption{The construction of the graph $G_Z$.  Self-loops omitted for clarity.}
\label{DetQuery}
\end{figure}

\section{Eulerian tour}

An \emph{Eulerian tour} in a graph $G$ is a cycle that goes through each edge of the graph exactly once.
A well-known theorem of Euler states that $G$ has such a tour iff it is connected and all its vertices have even degree.
Denote the corresponding communication complexity problem for $n$-vertex graphs by $\euler_n$.
Note that when the sets $E_A$ and $E_B$ are allowed to overlap, deciding if the degree $\deg(v)$ of a fixed vertex $v\in[n]$ is even
is essentially equivalent to $\IP_{n-1}$, as follows. Let $x\in\01^{n-1}$ be the characteristic vector of the neighbors of $v$ in $E_A$,
and $y\in\01^{n-1}$ the same for $E_B$, then we have $\deg(v)=|x\vee y|=|x|+|y|-|x\wedge y|$.
Since Alice and Bob can send each other the numbers $|x|$ and $|y|$ using a negligible $\log n$ bits,
computing $deg(v)\bmod 2$ is essentially equivalent to computing $|x\wedge y|\bmod 2=\IP_{n-1}(x,y)$.

Now we show how to embed into $\euler_{3n+4}$ an $\OR_n$ of disjoint $\IP_n$'s.
As usual, we first explain the reduction in the query world. For $i \in [n]$, let $z^i \in \01^n$,
and suppose that we want to compute $\OR_n(|z^1| \bmod 2, \ldots, |z^n| \bmod 2)$. We construct a graph $G$ with
$n+2$ left vertices $\ell_i$ and $n+2$ right vertices $r_i$ for $0 \leq i \leq n+1$, and $n$ middle vertices
$m_i$ for $i\in [n]$. Independently from the strings $z^i$, the graph
$G$ always has the edges $\{\ell_i, \ell_{i+1}\}$ and $\{r_i, r_{i+1}\}$ for $0 \leq i \leq n$
and the edges $\{m_i, m_{i+1}\}$ for  $1 \leq i \leq n-1$. It also contains the following 5 edges:
$\{\ell_0, m_1\},  \{r_0, m_1\}, \{\ell_{n+1}, m_n\}, \{r_{n+1}, m_n\}, \{m_1, m_n\}.$
We call these edges {\em fixed} edges. Finally, for every $(i,j)$ with
$z^i_j = 1$ we add the edges $\{\ell_i, m_j\}$ and $\{r_i, m_j\}$. Observe that $G$ is already connected by the fixed edges.  
See Figure~\ref{fig:euler} for an illustration.

\begin{myclaim}
\label{eulerclaim}
$G$ is Eulerian if and only if $\OR_n(|z^1| \bmod 2, \ldots, |z^n| \bmod 2)=0$.
\end{myclaim}

\begin{proof}
In the subgraph restricted to the fixed edges every vertex has even degree. Therefore
we can restrict our attention to the degrees with respect to the remaining edges that depend on the values $z^i_j$.
All the middle vertices have even degrees since for all $(i,j)$, we add 0 or 2 edges adjacent to $m_j$. For every $i \in [n],$
the degrees of $\ell_i$ and $r_i$ are the same since we add the edge $\{\ell_i, m_j\}$
exactly when we add the edge $\{r_i, m_j\}$. The degree of $\ell_i$ is the Hamming weight of $z^i$.
Therefore $G$ is Eulerian iff $|z^i|$ is even for all $i\in [n]$.
\end{proof}

\begin{figure}[hbt]
\centering
\includegraphics[scale=.6]{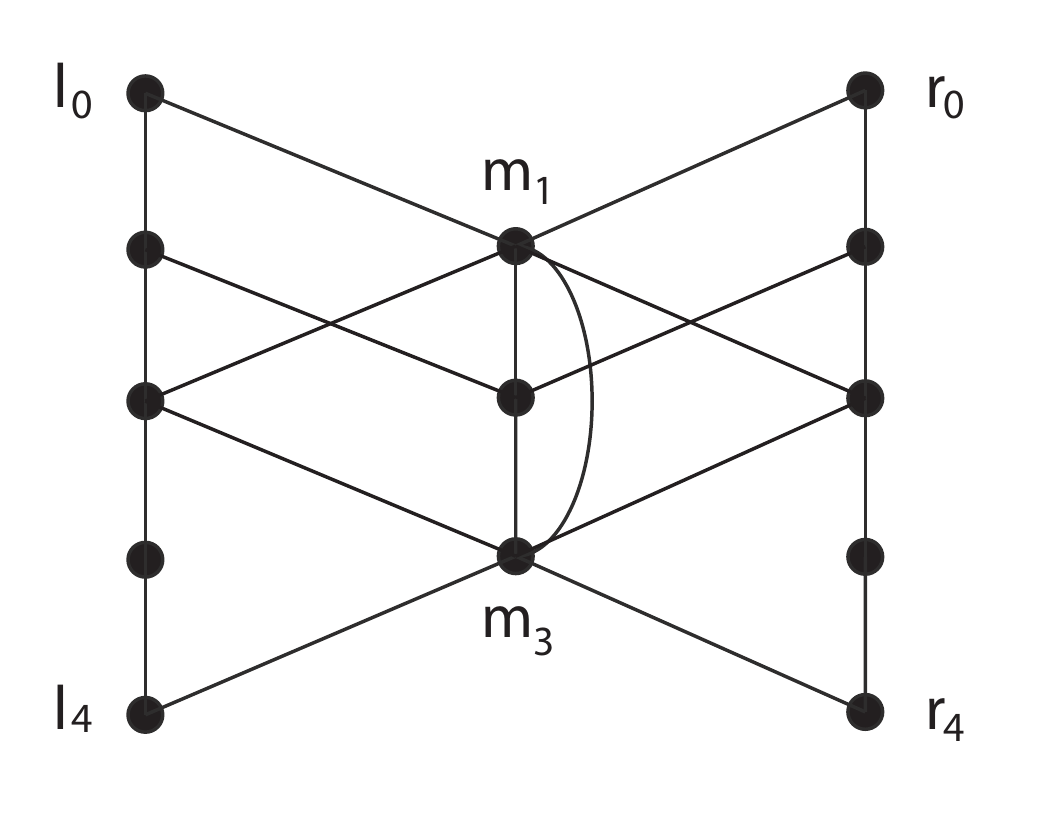}
\caption{Illustration of the graph to reduce OR of parities to Eulerian tour in the query model.
In this example, $n=3$ and $z^1=010, z^2=101,z^3=000$.}
\label{fig:euler}
\end{figure}

The transfer of this reduction to the communication complexity setting is
quite simple. Suppose that for each $i \in [n]$ Alice has string $x^i \in \01^n$, and Bob has 
$y^i \in \01^n$, and they want to compute the function $\OR_n(\IP_n(x^1,y^1), \ldots , \IP_n(x^n,y^n))$.
Let us suppose that $n$ is even, then $\IP_n(x^i,y^i) = \sum_j (\bar{x}^i_j \vee \bar{y}^i_j)$.
For all $(i,j)$ such that $x^i_j = 0$ we put the edges $\{\ell_i, m_j\}$ and $\{r_i, m_j\}$ in $E_A$,
and similarly, for $y^i_j = 0$ we put the edges $\{\ell_i, m_j\}$ and $\{r_i, m_j\}$ in $E_B$.
Thus in $E_A \cup E_B$ the edges $\{\ell_i, m_j\}$ and $\{r_i, m_j\}$ exist if and only if
$\bar{x}^i_j \vee \bar{y}^i_j = 1$. Therefore, by Claim~\ref{eulerclaim}
$\OR_n(\IP_n(x^1,y^1),\ldots,\IP_n(x^n,y^n)) = 0$ if and only if $G$ is Eulerian.

We can easily reduce $\DISJ_{n^2}$ on $n^2$-bit instances with intersection size~0 or~1 to $\OR_n\circ\IP_n$.
Since even that special case of $\DISJ_{n^2}$ requires linear classical communication~\cite{razborov:disj},
we obtain a tight lower bound $R_2(\euler_n)=\Omega(n^2)$.

The quantum communication complexity of
$\OR_n(\IP_n(x^1,y^1), \ldots , \IP_n(x^n,y^n))$ is $\Omega(n^{3/2})$.  This follows because
for any $f(g(x^1,y^1), \ldots, g(x^n,y^n))$ where $g$ is strongly balanced 
(meaning that all rows and columns in the communication matrix $M(x,y)=(-1)^{g(x,y)}$ sum to zero), 
the quantum communication complexity of $f$ is at least the approximate polynomial degree of $f$, 
times the discrepancy bound of $g$~\cite[Cor.~3]{lz:comp}.  
In our case, $\OR_n$ has approximate degree $\Omega(\sqrt{n})$ and $\IP_n$ contains 
a $2^{n-1}$-by-$2^{n-1}$ strongly balanced submatrix with discrepancy bound $\Omega(n)$.  
Thus we get $Q_2(\euler_n)\geq Q_2(\OR_n\circ\IP_n)\geq \Omega(n^{3/2})$.

This quantum lower bound is in fact tight: we first decide if $G$ is connected
using $O(n\log n)$ bits of communication (Section~\ref{secconnect}),
and if so then we use the Aaronson-Ambainis protocol to search for 
a vertex of odd degree (deciding whether a given vertex
has odd degree can be done deterministically with $O(n)$ bits of communication).
Thus we have:

\begin{theorem}
$R_2(\euler_n)=\Theta(n^2)$
and
$Q_2(\euler_n)=\Theta(n^{3/2})$.
\end{theorem}

\section{Other problems}

In this section we look at the quantum and classical communication complexity of a number of other graph properties.
Most results here are easy observations based on previous work, but worth making nonetheless.

\subsection{Triangle-finding}

Suppose we want to decide whether $G$ contains a triangle.
Papadimitriou and Sipser~\cite[pp.~266--7]{papadimitriou&sipser:cc}\footnote{Word of warning: Papadimitriou and Sipser \cite{papadimitriou&sipser:cc}
use the term ``inner product'' for what is now commonly called the ``intersection problem,'' i.e., the negation of disjointness.}
gave a reduction from $\DISJ_m$ to $\trianglefind_n$ for $m=\Omega(n^2)$,
which implies $R_2(\trianglefind_n)=\Theta(n^2)$.
Since we know that $Q_2(\DISJ_m)=\Theta(\sqrt{m})$, it also follows that $Q_2(\trianglefind_n)=\Omega(n)$.

This quantum lower bound is actually tight, which can be seen as follows.
First Alice checks if there already is a triangle within the edges $E_A$, and Bob does the same for $E_B$.
If not, then Alice defines the set of edges
$S_A=\{(a,b)\mid \exists~c~s.t.~(a,c),(b,c)\in E_A\}$ which would complete a
triangle for her, and uses the Aaronson-Ambainis protocol to try to find one among Bob's edges
(i.e., she searches for an edge in $S_A\cap E_B$).
Since $|S_A|\leq{n\choose 2}$, this process will find a triangle if Alice already
holds two of its edges, using $O(n)$ qubits of communication.
Bob does the same from his perspective.
If $G$ contains a triangle, then either Alice or Bob has at least two edges
of this triangle.  Hence this protocol will find a triangle with high probability
if one exists, using $O(n)$ qubits of communication.
Thus we have:

\begin{theorem}
$R_2(\trianglefind_n)=\Theta(n^2)$
and
$Q_2(\trianglefind_n)=\Theta(n)$.
\end{theorem}


\subsection{Bipartiteness}

Deterministic protocols can decide whether a given graph $G$ is bipartite using $O(n\log n)$ bits of communication, as follows.
Being bipartite is equivalent to being 2-colorable. Alice starts with some vertex $v_1$, colors it red,
and colors all of its neighbors (within $E_A$) blue.  Then she communicates all newly-colored vertices and their colors to Bob.
Bob continues coloring the neighbors of $v_1$ blue, and once he's done he communicates the newly-colored vertices and their colors to Alice.
If all vertices have been colored then Alice stops, otherwise she chooses an uncolored neighbor $v_2$ of a blue vertex, colors $v_2$ red,
and continues as above coloring $v_2$'s neighbors blue.
A connected graph is 2-colorable iff this process terminates without encountering a vertex colored both red and blue
(if the graph is not connected then Alice and Bob can treat each connected component separately).
Since each vertex will be communicated at most once, the whole process takes $O(n\log n)$ bits.

Babai et al.~\cite[Section~9]{bfs:classes} state a reduction from $\IP_n$ to bipartiteness,
which implies a near-matching quantum lower bound $Q_2(\bipart_n)=\Omega(n)$.

\begin{theorem}
$\Omega(n)\leq  Q_2(\bipart_n)\leq D(\bipart_n) \leq O(n\log n)$.
\end{theorem}

\subsection{Planarity}

Duris and Pudl\'{a}k~\cite{pudlak&duris:planarity} exhibited a deterministic protocol with $O(n\log n)$ bits
of communication for deciding whether a graph $G$ is \emph{planar} (i.e., whether it can be drawn in the plane without intersecting edges).
Babai et al.~\cite[Section~9]{bfs:classes} state a reduction from $\IP_n$ to planarity,
which implies a near-matching quantum lower bound $Q_2(\planar_n)=\Omega(n)$.

\section{Conclusion and open problems}

In this paper we studied the communication complexity (quantum and classical)
of a number of natural graph properties, obtaining nearly tight bounds for many of them.
We mention a few open problems:
\begin{itemize}
\item For \connect, can we improve the quantum upper bound from the trivial
$O(n\log n)$ to $O(n)$, matching the lower bound?
One option would be to run a distributed version of the $O(n)$-query
quantum algorithm of D\"urr et al.~\cite{dhhm:graphproblemsj},
but this involves a classical preprocessing phase that
seems to require $O(n\log n)$ communication.
Another option would be to run some kind of quantum random walk on the graph, starting from a random vertex,
and test whether it converges to a superposition of all vertices.
\item For \matching, can we show that the deterministic $O(n^{3/2}\log n)$-bit
protocol is essentially optimal, for instance by means of a $2^{\Omega(n^{3/2})}$
lower bound on the rank of the associated communication matrix?
Can we improve this upper bound using randomization and/or quantum communication,
possibly matching the $Q_2(\matching_n)=\Omega(n)$ lower bound?
\item Can we extend the $D(\matching_n) \leq O(n^{3/2}\log n)$ bound to general graphs?
\end{itemize}

\section*{Acknowledgements}
We thank Rahul Jain for several insightful discussions.

\bibliographystyle{alpha}
\bibliography{../qc}

\end{document}